\documentclass[oribibl,orivec]{llncs}

\usepackage{hyperref}
\usepackage{amsmath}
\usepackage{amssymb}
\usepackage{amsfonts}
\usepackage{graphicx}
\usepackage{mathrsfs}

\let\doendproof\endproof
\renewcommand\endproof{\hfill$\qed$\doendproof}

\newcommand{\overbar}[1]{\mkern 1.5mu\overline{\mkern-1.5mu#1\mkern-1.5mu}\mkern 1.5mu}
\newcommand{\Comp}[1]{\ensuremath{\overbar{#1}}}
\newcommand{\Rev}[1]{\ensuremath{\widetilde{#1}}}
\newcommand{\Back}[1]{\ensuremath{\widehat{#1}}}

\newcommand{\Factor}{\ensuremath{\preceq}}
\newcommand{\Prefix}{\ensuremath{\preceq_{\rm{pre}}}}
\newcommand{\Suffix}{\ensuremath{\preceq_{\rm{suff}}}}
\newcommand{\Affix}{\ensuremath{\preceq_{\rm{aff}}}}
\newcommand{\Middle}{\ensuremath{\preceq_{\rm{mid}}}}

\newcommand{\Mirror}{\ensuremath{\preceq_{\rm{mir}}}}
\newcommand{\Admissible}{\ensuremath{\preceq_{\rm{adm}}}}

\newcommand{\Bou}[1]{\ensuremath{\mathcal{B}(#1)}}

\newcommand{\Up}{\textbf{u}}
\newcommand{\Right}{\textbf{r}}
\newcommand{\Down}{\textbf{d}}
\newcommand{\Left}{\textbf{l}}

\begin{document}

\title{An Optimal Algorithm for Tiling the Plane\\with a Translated Polyomino} 
\titlerunning{Tiling with a Translated Polyomino}

\author{Andrew Winslow}
\authorrunning{Andrew Winslow}
\tocauthor{Andrew Winslow}

\institute{Universit\'{e} Libre de Bruxelles, 1050 Bruxelles, Belgium,\\
\email{andrew.winslow@ulb.ac.be}}

\maketitle

\begin{abstract}
We give a $O(n)$-time algorithm for determining whether translations of a polyomino with $n$ edges can tile the plane.
The algorithm is also a $O(n)$-time algorithm for enumerating all regular tilings, and we prove that at most $\Theta(n)$ such tilings exist.
\end{abstract}

\section{Introduction}

A \emph{plane tiling} is a partition of the plane into shapes each congruent to a fixed set of \emph{tiles}. 
As the works of M. C. Escher attest, plane tilings are both artistically beautiful and mathematically interesting (see~\cite{Schattschneider-1990} for a survey of both aspects).
In the 1960s, Golomb~\cite{Golomb-1965} initiated the study of \emph{polyomino} tiles: polygons whose edges are axis-aligned and unit-length.

Building on work of Berger~\cite{Berger-1966}, Golomb~\cite{Golomb-1970} proved that no algorithm exists for determining whether a set of polyomino tiles has a plane tiling.
Ollinger~\cite{Ollinger-2009} proved that this remains true even for sets of at most 5~tiles.
It is a long-standing conjecture that there exists an algorithm for deciding whether a single tile admits a plane tiling (see~\cite{Strauss-2000,Strauss-2010}) 

Motivated by applications in parallel computing, Shapiro~\cite{Shapiro-1978} studied tilings of polyomino tiles on a common integer lattice using translated copies of a polyomino.
For the remainder of the paper, only these tilings are considered.
Ollinger~\cite{Ollinger-2009} proved that no algorithm exists for determining whether sets of at most 11~tiles admit a tiling, while Wijshoff and van Leeuwen~\cite{Wijshoff-1984} obtained a polynomial-time-testable criterion for a single tile to admit a tiling.
Beauquier and Nivat~\cite{Beauquier-1991,Girault-1990} improved on the result of Wijshoff and van Leeuwen by giving a simpler criterion called the \emph{Beauquier-Nivat criterion}.

Informally, a tile satisfies the Beauquier-Nivat criterion if it can be surrounded by copies of itself (see Figure~\ref{fig:surrounding}).
Such a surrounding must correspond to a \emph{regular} tiling (also called \emph{isohedral}) in which all tiles share an identical neighborhood.
Using a naive algorithm, the Beauquier-Nivat criterion can be applied to a polyomino with $n$ vertices in $O(n^4)$ time. 

\begin{figure}[ht]
\centering
\includegraphics[scale=1.0]{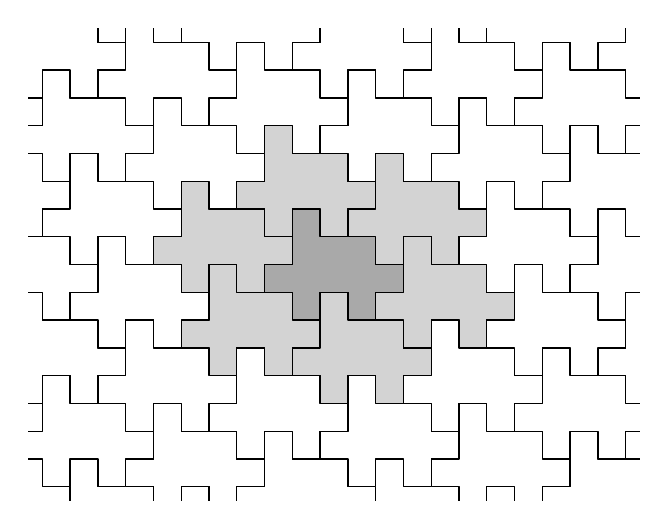}
\caption{A polyomino tile (dark gray), a surrounding of the tile (gray), and the induced regular tiling (white).}
\label{fig:surrounding}
\end{figure}

The $O(n^4)$ algorithm of~\cite{Beauquier-1991} is implicit; the main achievement of~\cite{Beauquier-1991} is a concise characterization of exact tiles, akin to Conway's criterion (see~\cite{Schattschneider-1980}).
Gambini and Vuillon~\cite{Gambini-2007} gave an improved $O(n^2)$-time algorithm utilizing structural and algorithmic results on words describing boundaries of polyominoes.
Around the same time, Brlek, Proven\c{c}al, and F\'{e}dou~\cite{Brlek-2006,Brlek-2009a} also used a word-based approach to achieve $O(n)$-time algorithms for two special cases: (1) the boundary contains no consecutive repeated sections larger than $O(\sqrt{n})$, and (2) testing a restricted version of the Beauquier-Nivat criterion (surroundable by just four copies).
Proven\c{c}al~\cite{Provencal-2008} further improved on the algorithm of Gambini and Vuillon for the general case, obtaining $O(n\log^3(n))$ running time.
In a recent survey of the combinatorics of Escher's tilings, Blondin Mass\'{e}, Brlek, and Labb\'{e}~\cite{Blondin-2010b} conjecture that a $O(n)$-time algorithm exists.
In this work, we confirm their conjecture by giving such an algorithm (Theorem~\ref{thm:linear-time-iso-enum}).

The algorithm doubles as an algorithm for enumerating all surroundings (regular tilings) of the polyomino.
As part of the proof of the algorithm's running time, we prove a claim of Proven\c{c}al~\cite{Provencal-2008} that the number of surroundings of a tile with itself is $O(n)$ (Corollary~\ref{cor:linear-tilings}).
This complements the tight bounds on a special class of surroundings by Blondin Mass\'{e} et al.~\cite{Blondin-2009,Blondin-2010a}, and proves that our $O(n+k)$-time algorithm for enumerating all $k$ surroundings (Lemma~\ref{lem:translation-linear-enumerate}) is also a $O(n)$-time algorithm.

\section{Definitions}
\label{sec:definitions}

Here we give precise formulations of terms used throughout the paper. 
The definitions are similar to those of Beauquier and Nivat~\cite{Beauquier-1991} and Brlek et al.~\cite{Brlek-2009a}.

\subsection{Words}

A \emph{letter} is a symbol $x \in \Sigma = \{\Up, \Down, \Left, \Right\}$.
The \emph{complement} of a letter $x$, written $\Comp{x}$, is defined by the following bijection on $\Sigma$: $\Comp{\Up} = \Down$, $\Comp{\Right} = \Left$, $\Comp{\Down} = \Up$, and $\Comp{\Left} = \Right$. 

A \emph{word} is a sequence of letters and the \emph{length} of a word $W$, denoted $|W|$, is the number of letters in $W$.
For an integer $i \in \{1, 2, \dots, |W|\}$, $W[i]$ refers to the $i$th letter of $W$ and $W[-i]$ refers to the $i$th from the last letter of $W$.
The notation $l^k$ or $W^k$ denotes the word consisting of $k$ repeats of a letter $l$ or word $W$, respectively.

There are several functions mapping a word $W$ to another word of the same length.
The \emph{complement} of $W$, written $\Comp{W}$, is the word obtained by replacing each letter of $W$ with its complement.
The \emph{reverse} of $W$, written $\Rev{W}$, are the letters of $W$ in reverse order.
The \emph{backtrack} of $W$, written $\Back{W}$, is defined as $\Back{W} = \Comp{\Rev{W}}$.
Note that for any two words $X$ and $Y$, $\Back{AB} = \Back{B}\Back{A}$.

\subsection{Factors}

A \emph{factor of $W$} is an occurrence of a word in $W$, written $X \Factor W$.
For integers $1 \leq i, j \leq |W|$ with $i \leq j$, $W[i..j]$ denotes the factor of $W$ from $W[i]$ to $W[j]$, inclusive.
A factor $X$ \emph{starts} or \emph{ends} at $W[i]$ if $W[i]$ is the first or last letter of $X$, respectively.

Two factors $X, Y \Factor W$ may refer the same letters of $W$ or merely have the same letters in common.
In the former case, $X$ and $Y$ are \emph{equal}, written $X = Y$, while in the latter, $X$ and $Y$ are \emph{congruent}, written $X \equiv Y$.
For instance, if $W = \Up \Up \Up \Left \Right \Up \Up \Up$ then $W[1..3] \equiv W[6..8]$.
A \emph{factorization} of $W$ is a partition of $W$ into consecutive factors $F_1$ through $F_k$, written $W = F_1 F_2 \dots F_k$.

\subsection{Special words and factors}

A word $X$ is a \emph{prefix} or \emph{suffix} of a word $W$ provided $W = XU$ or $W = UX$, respectively. 
A word \emph{$X$ is a period of $W$} provided $|X| \leq |W|$ and $W$ is a prefix of $X^k$ for some $k \geq 1$ (introduced in~\cite{Knuth-1977}).
Alternatively, $X$ is a prefix of $W$ and $W[i] = W[i+|X|]$ for all $1 \leq i \leq |W|-|X|$.

A factor $X \Factor W$ is a \emph{prefix} if $X$ starts at $W[1]$, written $X \Prefix W$.
Similarly, $X \Factor W$ is a \emph{suffix} if $X$ ends at $W[-1]$, written $X \Suffix W$.
A factor $X \Factor W$ that is either a prefix or suffix is an \emph{affix}, written $X \Affix W$.
A factor $X \Factor W$ that is not an affix is a \emph{middle}, written $X \Middle W$.

The factor $X \Factor W$ such that $W = UXV$, $|U|=|V|$, and $|X| \in \{1, 2\}$ is the \emph{center of $W$}.
A factor $X \Factor W$ is a \emph{mirror}, written $X \Mirror W$, provided $W = XUYV$ with $Y \equiv \Back{X}$ and $|U|=|V|$.
For any $X \Mirror W$, $\Back{X}$ refers to the factor $Y$ in the definition.

A mirror factor is \emph{admissible} provided $U[1] \neq \Comp{U[-1]}$, $V[1] \neq \Comp{V[-1]}$.
Observe that each admissible factor is the maximum-length mirror factor with its center.
Thus any two admissible factors have distinct centers.

\subsection{Polyominoes and boundary words}

A \emph{cell} is a unit square with lower-leftmost vertex $(x, y) \in \mathbb{Z}^2$ and remaining vertices $(x+1, y)$, $(x, y+1)$, $(x+1, y+1)$.
A \emph{polyomino} is a simply connected union of cells whose boundary is a simple closed curve.

The boundary of a polyomino consists of cell edges.
The \emph{boundary word} of a polyomino $P$, denoted $\Bou{P}$, is the circular word of letters corresponding to the sequence of directions traveled along cell edges during a clockwise traversal of the polyomino's boundary (see Figure~\ref{fig:regularity}).

Boundary words are \emph{circular}: the last and first letters are defined to be consecutive.
Thus for any indices $i, j \in \mathbb{Z} \setminus \{0\}$, $W[i]$ and $W[i..j]$ are defined.
For the boundary word $W = \Up \Right \Right \Down \Left \Left$, $W[10] = W[-9] = \Down$ and $W[6..2] = \Left \Up \Right$.

\begin{figure}[ht]
\centering
\includegraphics[scale=1.0]{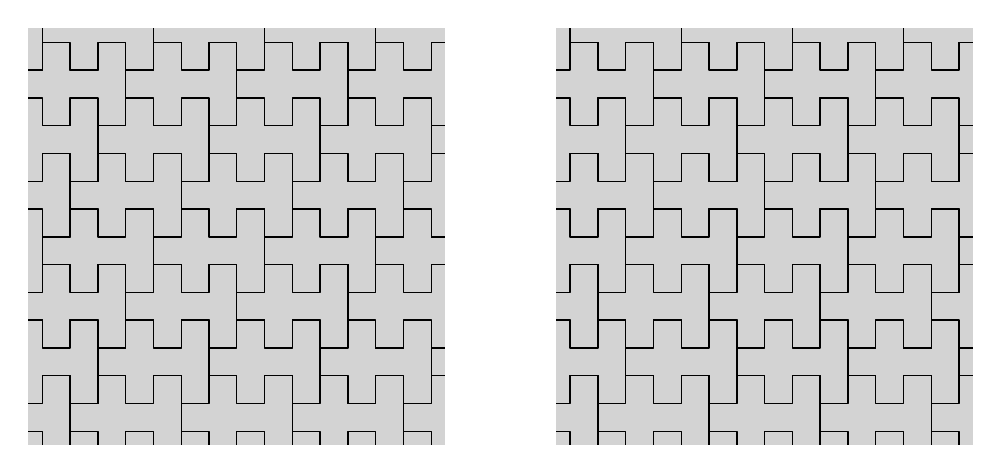}
\caption{A regular tiling (left) and non-regular tiling (right) of a polyomino with boundary word $\Up \Right \Up \Right \Down \Right \Up \Right \Down^3 \Left \Up \Left \Down \Left \Up \Left$. 
The copies in the regular tiling have a common neighborhood factorization $ABC\Back{A}\Back{B}\Back{C}$, with $A = \Up$, $B = \Right \Up$, $C = \Right \Down \Right \Up \Right \Down$.
}
\label{fig:regularity}
\end{figure}

\subsection{Tilings}

For a polyomino $P$, a \emph{tiling} of $P$ is an infinite set $\mathscr{T}$ of translations of $P$, called \emph{copies}, such that every cell in the plane is in exactly one copy.
A tiling is \emph{regular} (e.g. \emph{isohedral}) provided there exist vectors $\vec{o}, \vec{u}, \vec{v}$ such that the set of lower-leftmost vertices of copies in the tiling is $\vec{o} + \{i\vec{u} + j\vec{v} : i, j \in \mathbb{Z}\}$.
Two tilings $\mathcal{T}$ and $\mathcal{T}'$ are equal provided there exists a vector $\vec{v}$ such that $\mathcal{T}' = \vec{v} + \mathcal{T}$.

Copies of a tiling intersect only along boundaries, and copies with non-empty boundary intersection are \emph{neighbors}.
Lemma 3.5 of~\cite{Wijshoff-1984} implies that the intersection between a pair of neighbors corresponds to a \emph{neighbor factor} of each neighbor's boundary word and these factors form a \emph{neighborhood factorization}.
Every regular tiling has a neighbor factorization common to all copies in the tiling.

\section{The Beauquier-Nivat Criterion}
\label{sec:translation-crit}

Recall that $\Back{X}$ is the reverse complement of $X$.
Thus $\Back{X}$ is the same path as $X$ but traversed in the opposite direction.
So any pair of factors $X$ and $\Back{X}$ appearing on the boundary of a polyomino are translations of each other with the interior of the boundary on opposite sites.
Beauquier and Nivat~\cite{Beauquier-1991} gave the following criterion for determining whether a polyomino tile admits a tiling: 

\begin{definition}
A factorization $W = ABC\Back{A}\Back{B}\Back{C}$ of a boundary word $W$ is a \emph{BN factorization}.
\end{definition}

\begin{lemma}[Theorem 3.2 of~\cite{Beauquier-1991}]
\label{lem:tiling-translation}
A polyomino $P$ has a tiling if and only if $\Bou{P}$ has a BN factorization.
\end{lemma}

As seen in Figure~\ref{fig:translation-ex}, a BN factorization corresponds to the neighborhood factorization of a regular tiling. 
We prove this formally by reusing results from the proof of Lemma~\ref{lem:tiling-translation}.

\begin{figure}[ht]
\centering
\includegraphics[scale=1.0]{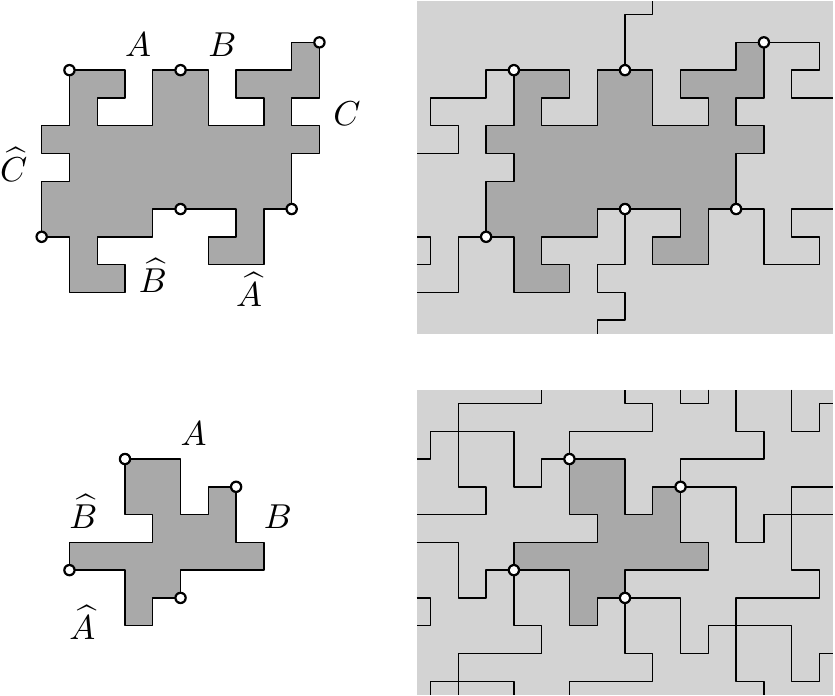}
\caption{BN factorizations (left) and the regular tilings induced by these factorizations (right).
For one polyomino (bottom), two of the factors are zero length.
However, no BN factorization can have more than two length-0 factors.}
\label{fig:translation-ex}
\end{figure}

\begin{lemma}[Corollary 3.2 of~\cite{Beauquier-1991}]
\label{lem:exact-contact}
Let $P$ be a polyomino.
There exists a factorization $\Bou{P} = F_1 \Back{F_3} F_2 \Back{F_1} F_3 \Back{F_2}$ if and only if there exists a tiling $\mathscr{T}$ of $P$ with three copies $P_1$, $P_2$, $P_3$ such that:
\begin{itemize}
\item $P_1, P_2, P_3$ appear clockwise consecutively around a common point $q$.
\item $F_i$ is the last neighbor factor of $P_i$ whose clockwise endpoint is incident to $q$.
\end{itemize}
\end{lemma}

\begin{lemma}
\label{lem:translation-bijection}
Let $P$ be a polyomino.
A factorization of $\Bou{P}$ is a BN factorization if and only if a regular tiling of $P$ has this neighbor factorization.  
\end{lemma}

\begin{proof}
The factorization $\Bou{P} = F_1 \Back{F_3} F_2 \Back{F_1} F_3 \Back{F_2}$ is a generic BN factorization.
So it suffices to prove that there exists a tiling $\mathscr{T}$ of $P$ satisfying the conditions of Lemma~\ref{lem:exact-contact} if and only if there exists a regular tiling $\mathscr{T}_{\rm{reg}}$ with neighbor factorization $\Bou{P} = F_1 \Back{F_3} F_2 \Back{F_1} F_3 \Back{F_2}$.

\textbf{Tiling $\Rightarrow$ neighbor factorization.}
Let $\mathscr{T}$ be a tiling and $P_1, P_2, P_3 \in \mathscr{T}$ be copies as defined in the statement of Lemma~\ref{lem:exact-contact}.
Let $\vec{u}$ and $\vec{v}$ be the amount $P_2$ and $P_3$ are translated relative to $P_1$, respectively.
Lemma 3.2 of~\cite{Beauquier-1991} states that the copies obtained by translating $P_1$ by $\vec{u}, \vec{v}, \vec{v}-\vec{u}$, $-\vec{u}$, $-\vec{v}$, and $\vec{u}-\vec{v}$ is a \emph{surrounding} of $P_1$: a set of interior-disjoint copies such that every edge of $C_1$ is shared by a copy.
Since $P_3$ is a copy of $P_2$ translated by $\vec{v}-\vec{u}$, the neighbor factor of $P_1$ incident to the copy translated by $\vec{v}-\vec{u}$ is $F_2$.
By similar reasoning, $P_1$ has neighbor factors $\Back{F_1}$, $F_3$, and $\Back{F_2}$ incident to the copies translated by $-\vec{u}$, $-\vec{v}$, and $\vec{u}-\vec{v}$, respectively.
So $P_1$ has neighbor factorization $\Bou{P} = F_1 \Back{F_3} F_2 \Back{F_1} F_3 \Back{F_2}$.
Corollary 3.1 of~\cite{Beauquier-1991} states that for every surrounding, there exists a regular tiling of $P$ containing the surrounding and thus has the neighbor factorization of $P_1$. 

\textbf{Tiling $\Leftarrow$ neighbor factorization.}
Now suppose there exists a regular tiling $\mathscr{T}_{\rm{reg}}$ of $P$ with neighbor factorization $F_1 \Back{F_3} F_2 \Back{F_1} F_3 \Back{F_2}$.
Let $P_1 \in \mathscr{T}_{\rm{reg}}$ be a copy and $q$ be the clockwise endpoint of the factor $F_1$ of $P_1$.
Let $P_2, P_3 \in \mathscr{T}_{\rm{reg}}$ be copies adjacent to $P_1$ and incident to factors $F_1$ and $\Back{F_3}$ of $P_1$.
Let $\vec{u}$ and $\vec{v}$ be the amount $P_2$ and $P_3$ are translated relative to $P_1$, respectively.
Then $q$ is the clockwise endpoint of the factor $F_2$ of $P_1$, translated by $\vec{u}$.
Also, $q$ is the clockwise endpoint of the factor $F_3$ translated by, translated by $\vec{v}$.
So the factors of $P_2$ and $P_3$ whose clockwise endpoints are $q$ are $F_2$ and $F_3$, respectively.
\end{proof}

\section{A Bound on the Number of Factorizations}
\label{sec:combinatorial}

Here we prove that the number of BN factorizations of the boundary word of an $n$-omino is $O(n)$.
This fact is used in Section~\ref{sec:combinatorial} to improve the bound on the running time the algorithm from $O(n+k)$ to $O(n)$.

\begin{lemma}
\label{lem:presuf-period}
Let $W$ be a boundary word with a factor $X$.
Let $P,S \Mirror W$ such that $P \Prefix X$, $S \Suffix X$, and $P \neq S$.
Then $X$ has a period of length $2|X|-(|P|+|S|)$. 
\end{lemma}

\begin{proof}
Since $P$ and $S$ are mirror, there exists $X' \Factor W$ with $|X'| = |X|$, $\Back{P} \Prefix X'$, and $\Back{S} \Suffix X'$.
Observe that $X$ has a period of length $r \geq 1$ if and only if $X[i] = X[i+r]$ for all $1 \leq i \leq |X|-r$.
Let $1 \leq i \leq |P|+|S|-|X|$.
Then $1 \leq |P|+1-i \leq |X|$ and $1 \leq |P|+1+|\Back{S}|-|X'|-i \leq |\Back{S}|$.
So:
\begin{equation*}
\begin{split}
X[i] &= P[i] \\
&= \Comp{\Back{P}}[|P|+1-i] \\
&= \Comp{X'}[|P|+1-i] \\
&= \Comp{\Back{S}}[|P|+1+|\Back{S}|-|X'|-i] \\
&= \Comp{\Back{S}}[|\Back{S}|+1-(i+|X'|-|P|)] \\
&= S[i+|X'|-|P|] \\
&= X[i+|X'|-|P|+(|X|-|S|)] \\
&= X[i+2|X|-(|P|+|S|)]
\end{split}
\end{equation*}
Since $P \neq S$, $2|X|-(|P|+|S|) \geq 2|X|-(2|X|-1) = 1$.
So $X$ has a period of length $2|X|-(|P|+|S|)$.
\end{proof}

\begin{lemma}
\label{lem:long-middle-nonadmissible}
Let $W$ be a boundary word with $X \Factor W$.
Let $P, S \Mirror W$ such that $P \Prefix X$, $S \Suffix X$, and $P \neq S$.
Any factor $Y \Middle X$ with $|Y| > 2|X|-(|P|+|S|)$ is not an admissible factor of $W$. 
\end{lemma}

\begin{proof}
By Lemma~\ref{lem:presuf-period}, $X$ has a period of length $r = 2|X|-(|P|+|S|)$.
Let $Y \Middle X$ and $|Y| > r$.

Let $X' \Factor W$ with $|X'| = |X|$ and the center of $X'$ exactly $|W|/2$ letters from the center of $X$.
Then $\Back{P} \Prefix X'$, $\Back{S} \Suffix X'$, and $\Back{Y} \Middle X'$.
Again by Lemma~\ref{lem:presuf-period}, $X'$ has a period of length $r$.

Let $U, V \Factor W$ such that $W = YU\Back{Y}V$.
Since $Y$ is a middle factor of $X$, the letter $U[1]$ is in $X$.
Since $X$ has a period of length $r$ and $|Y| > r$, $U[1] = Y[|Y|+1-r] = \Comp{\Back{Y}[r]}$.
Since $\Back{Y}$ is a middle factor of $X'$ and $X'$ has a period of length $r$, $U[-1] = \Back{Y}[r]$.
So $U[1] = \Comp{U[-1]}$ and $Y$ is not admissible.
\end{proof}

\begin{lemma}
\label{lem:constant-big-admissible}
Let $W$ be a boundary word.
There exists a set $\mathscr{F}$ of $O(1)$ factors of $W$ such that every $F \Admissible W$ with $|F| \geq |W|/6$ is an affix factor of an element of $\mathscr{F}$. 
\end{lemma}

\begin{proof}
\textbf{A special case on three factors.}
Let $P_1, P_2, P_3 \Admissible W$ with $|P_1|, |P_2|, |P_3| \geq |W|/6$ and centers contained in a factor of $W$ with length at most $|W|/14$.
Let $X \Factor W$ be the shortest factor such that $P_1, P_2, P_3 \Factor X$, and so $P_i \Prefix X$ and $P_j \Suffix X$ for some $i, j \in \{1, 2, 3\}$. 
We prove that if $i \neq j$, then $P_1, P_2, P_3 \Affix X$.

Without loss of generality, suppose $i = 1$, $j = 2$ and so $P_3 \Middle X$.
By Lemma~\ref{lem:long-middle-nonadmissible}, since $P_3 \Admissible W$, $|P_3| \leq 2|X| -(|P_1|+|P_2|) \leq |P_1| + |W|/7 + |P_2| - (|P_1|+|P_2|) = |W|/7 < |W|/6$, a contradiction.
So $P_3 \Affix X$.

\textbf{All nearby factors.}
Consider a set $\mathscr{I} = \{F_1, F_2, \dots, F_m\}$ of at least three admissible factors of $W$ of length at least $|W|/6$ such that the centers of the factors are contained in a common factor of $W$ of length $|W|/14$.
We will prove that every element of $\mathscr{I}$ is an affix factor of one of two factors of $W$.

Let $G \Factor W$ be the shortest factor such that $F_i \Factor G$ for every $F_i \in \mathscr{I}$.
It is either the case that there exist distinct $F_l, F_r \in \mathscr{I}$ with $F_l \Prefix G$, $F_r \Suffix G$, or that $G \in \mathscr{I}$ and every $F_i \in \mathscr{I}$ besides $G$ has $F_i \Middle G$. 

In the first case, $F_i \Affix G$ for any $i \neq l, r$ by the previous claim regarding three factors.
Also $F_l, F_r \Affix G$.
So every factor in $\mathscr{I}$ is an affix factor of $G$.

In the second case, let $G' \Factor G$ be the shortest factor with the same center as $G$ such that every factor in $\mathscr{I}$ excluding $G$ is a factor of $G'$.
Clearly $G' \Mirror W$ and $G' \not \Admissible W$.
Without loss of generality, there exists $F_p \in \mathscr{I}$ such that $F_p \Prefix G'$.
Since $F_p \Admissible W$ and $G' \not \Admissible W$, $F_p \neq G'$.

Applying Lemma~\ref{lem:long-middle-nonadmissible} with $X = G'$, $P = F_p$, $S = G'$, every middle factor of $G'$ in $\mathscr{I}$ has length at most $2|G'|-(|G'|+|F_p|) \leq |G'| - |F_p| \leq |W|/7 < |W|/6$.
So every factor of $G'$ in $\mathscr{I}$ is an affix factor of $G'$.
Thus every factor in $\mathscr{I}$ is either $G$ or an affix factor of $G'$.

\textbf{All factors.}
Partition $W$ into 15 factors $I_1, I_2, \dots, I_{15}$ each of length at most $|W|/14$.
Let $\mathscr{I}_i$ be the set of admissible factors with centers containing letters in $I_i$.
Then by the previous claim regarding more than three factors, there exists a set $\mathscr{F}_i$ ($G$ and possibly $G'$) such that every element of $\mathscr{I}_i$ is an affix factor of an element of $\mathscr{F}_i$ and $|\mathscr{F}_i| \leq 2$.
So every $F \Admissible W$ with $|F| \geq |W|/6$ is an affix factor of an element of $\mathscr{F} = \bigcup_{i=1}^{15}{\mathscr{F}_i}$ and $|\mathscr{F}| \leq 2\cdot15$.
\end{proof}

\begin{theorem}
\label{thm:linear-factorizations}
A boundary word $W$ has $O(|W|)$ BN factorizations.
\end{theorem}

\begin{proof}
Consider the choices for the three factors $A$, $B$, $C$ of BN factorization $W = ABC\Back{A}\Back{B}\Back{C}$.
In any factorization, some factor has size at least $|W|/6$.
By Lemma~\ref{lem:constant-big-admissible}, there exists a $O(1)$-sized set of factors $\mathscr{F}$ such that any factor with length at least $|W|/6$ is an affix factor of an element of $\mathscr{F}$.
Without loss of generality, either $|A| \geq |W|/6$ and $A$ is a prefix of a factor in $\mathscr{F}$ or $|C| \geq |W|/6$ and $C$ is a suffix of a factor in $\mathscr{F}$.

Let $H = ABC$ be the factor formed by consecutive factors $A$, $B$, $C$ of a BN factorization.
Then since $|H| = |W|/2$ and shares either the first or last letter with a factor in $\mathscr{F}$, there are $O(1)$ total factors $H$.
For a fixed $H$, choosing the center of $B$ determines $B$ (since $B$ is admissible) and thus $A$ and $C$.
So there are at most $2(|W|/2)$ factorizations for a fixed factor $H$.
\end{proof}

Since Lemma~\ref{lem:translation-bijection} proves that factorizations and tilings are equivalent, the previous theorem implies a linear upper bound on the number of regular tilings of a polyomino: 

\begin{corollary}
\label{cor:linear-tilings}
An $n$-omino has $O(n)$ regular tilings.
\end{corollary}

As pointed out by Proven\c{c}al~\cite{Provencal-2008}, it is easy to construct polyominoes with $\Omega(n)$ such tilings.
For instance, the polyomino with boundary word $W = \Up \Right^i \Down \Left^i$ with $i \geq 1$ has $|W|/2-1$ regular tilings. 

\section{An Algorithm for Enumerating Factorizations}
\label{sec:algorithm}

The bulk of this section describes a $O(|W|)$-time algorithm for enumerating the factorizations of a polyomino boundary word $W$.
The algorithm combines algorithmic ideas of Brlek, Proven\c{c}al, and F\'{e}dou~\cite{Brlek-2009a} and a structural result based on a well-known lemma of Galil and Seirferas~\cite{Galil-1978}.   

\begin{lemma}[Corollary 5 of~\cite{Brlek-2009a}] 
\label{lem:translation-factors-maximal}
Every factor of a BN factorization is admissible.
\end{lemma}

Lemma~\ref{lem:extremal-factorizations-translation} is a variation of Lemma~C4 of Galil and Seirferas~\cite{Galil-1978}.
We reproduce their proof with minor modifications.

\begin{lemma}
\label{lem:extremal-factorizations-translation}
Let $A$ and $B$ be two words of the same length.
Moreover, let $A = X_1 X_2 = Y_1 Y_2 = Z_1 Z_2$ and $B = X_Q \Back{X_2} = \Back{Y_1} \Back{Y_2} = \Back{Z_1} Z_Q$ with $|X_1| < |Y_1| < |Z_1|$.
Then $X_Q = \Back{X_1}$ and $Z_Q = \Back{Z_2}$. 
\end{lemma}

\begin{proof}
Let $V$ be the word such that $Y_1 V = Z_1$ (see Figure~\ref{fig:galil-double-strings}).

\begin{figure}[ht]
\centering
\includegraphics[scale=1.0]{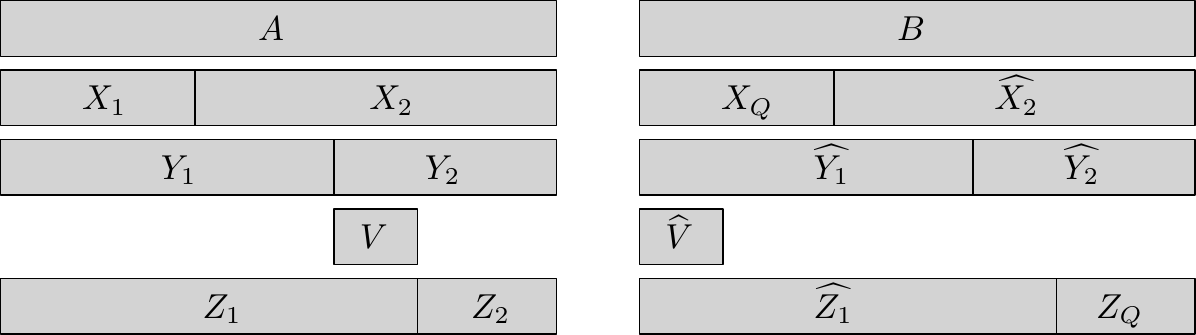}
\caption{The words used in the proof of Lemma~\ref{lem:extremal-factorizations-translation}.}
\label{fig:galil-double-strings}
\end{figure}

\textbf{Claim~(1): $\Back{V}$ is a period of $\Back{Z_1}$.}
Since $Y_1 V = Z_1$, then $\Back{Z_1} = \Back{Y_1 V} = \Back{V} \Back{Y_1}$ is a prefix of $B$.
So $\Back{Y_1}$ is a prefix of $\Back{Z_1} = \Back{V} \Back{Y_1}$ and thus $\Back{V}$ is a period of $\Back{Y_1}$.
So $\Back{V}$ is a period of $\Back{V} \Back{Y_1} = \Back{Z_1}$.

\textbf{Claim~(2): $V$ is a prefix of $X_2$.} 
Since $V$ is a prefix of $Y_2$, $\Back{V}$ is a suffix of $\Back{Y_2}$.
So $\Back{V}$ is a suffix of $\Back{X_2}$ and $V$ is a prefix of $X_2$. 

\textbf{Claim~(3): $X_1 V$ is a prefix of $Z_1$.}
Since $V$ is a prefix of $X_2$, $X_1 V$ is a prefix of $Y_1 V$.
Since $|X_1 V| < |Y_1 V| = |Z_1|$, $X_1 V$ is also a prefix of $Z_1$.

\textbf{Claim~(4): $\Back{V}$ is a period of $\Back{X_1}$.} 
By claim~(1), $\Back{V}$ is a period of $\Back{Z_1}$, so $Z_1$ has a period of length $|\Back{V}| = |V|$.
By claim~(3), $X_1 V$ is a prefix of $Z_1$ and so also has a period of length $|V|$.
Then $\Back{X_1 V} = \Back{V} \Back{X_1}$ has a period of length $|V|$, namely $\Back{V}$.
So $\Back{V}$ is also a period of $\Back{X_1}$. 
 
Finally, combining claims~(1) and~(4), since $\Back{V}$ is a period of both $X_Q$ and $\Back{X_1}$, $X_Q = \Back{X_1}$.
By symmetry, the same proof also implies $Z_Q = \Back{Z_2}$.
\end{proof}

\begin{lemma}[Theorem 9.1.1 of~\cite{Gusfield-1997}]
\label{lem:longest-common-extension}
Two non-circular words $X$, $Y$ can be preprocessed in $O(|X| + |Y|)$ time to support the following queries in $O(1)$-time:
what is the longest common factor of $X$ and $Y$ starting at $X[i]$ and $Y[j]$?
\end{lemma}

\begin{lemma}
\label{lem:translation-linear-enumerate}
Let $W$ be a polyomino boundary word.
Then the BN factorizations of $W$ can be enumerated in $O(|W|)$ time. 
\end{lemma}

\begin{proof}
Lemma~\ref{lem:translation-factors-maximal} states that BN factorizations consist entirely of admissible factors.
The algorithm first computes all admissible factors, then searches for factorizations consisting of them.

\textbf{Computing admissible factors.}
Lemma~\ref{lem:translation-factors-maximal} implies that there are at most $2|W|$ admissible factors, since admissible factor has a distinct center.
For each center $W[i..i]$ or $W[i..i+1]$, the admissible factor with this center is $LR$, where $R$ is the longest common factor of $W$ starting at $W[i+1]$ and $\Back{W}$ starting at $\Back{W}[|W|/2-(i+1)]$.
Similarly, $L$ is the longest common factor of $\Rev{W}$ starting at $\Rev{W}[|W|/2-i]$ and $\Comp{W}$ starting at $\Comp{W}[i]$. 
Preprocess $WW$, $\Back{W}\Back{W}$, $\Rev{W}\Rev{W}$, and $\Comp{W}\Comp{W}$ using Lemma~\ref{lem:longest-common-extension} so that each longest common factor can be computed in $O(1)$ time.
If $|L| \neq |R|$, then $X$ is not admissible and is discarded.
Since $O(1)$ time is spent for each of $2|W|$ admissible factors, this step takes $O(|W|)$ total time.

\textbf{Enumerating factorizations.}
Let $W = AY\Back{A}Z$ with $A$ an admissible factor and $|Y| = |Z|$.
Let $B_1, B_2, \dots, B_l$ be the admissible prefix factors of $Y$, with $|B_1| < |B_2| < \dots < |B_l|$.
Similarly, let $C_1, \dots, C_m$ be the suffix factors with $|C_1| < \dots < |C_m|$.
Lemma~\ref{lem:extremal-factorizations-translation} implies that for fixed $A$, there exist intervals $[b, l]$, $[c, m]$ such that the BN factorizations $AB_iC_j\Back{A}\Back{B_i}\Back{C_j}$ are exactly those with $i \in [b, l]$ or $j \in [c, m]$.

First, construct a length-sorted list of the admissible factors starting at each $W[k]$ in $O(|W|)$ time using counting sort.
Do the same for all factors ending at each $W[k]$.

Next, use a two-finger scan to find, for each factor $A$ that ends at $W[k]$, the longest factor $B_l$ starting at $W[k+1]$ such that $|A|+|B_l| \leq |W|/2$.
Then check whether $C_j$, the factor following $B_l$ such that $|AB_lC_j| = |W|/2$, is admissible and report the factorization $AB_lC_j\Back{A}\Back{B_l}\Back{C_j}$ if so.
Checking whether $C_j$ is admissible takes $O(1)$ time using an array mapping each center to the unique admissible factor with this center.

Additional BN factorizations containing $A$ are enumerated by checking factors $B_i$ with $i = l-1, l-2, \dots$ for an admissible following factor $C_j$.
Either $C_j$ is admissible and the factorization is reported, or $i = b-1$ and the iteration stops.

Finally, use a similar two-finger scan to find, for each factor $A$ that starts at $W[k]$, the longest factor $C_m$ that ends at $W[k+|W|/2-1]$ such that $|A|+|C_m| \leq |W|/2$, check whether $B_i$ preceeding $C_m$ such that $|AB_iC_m| = |W|/2$ is admissible, and report the possible BN factorization.
Then check and report similar factorizations with $C_j$ for $j = m-1, m-2, \dots$ until $j = c-1$.

In total, the two-finger scans take $O(|W|)$ time plus $O(1)$ time to report each factorization.
Reporting duplicate factorizations can be avoided by only reporting a factorization if $A[1]$ appears before $B[1]$, $C[1]$, $\Back{A}[1]$, $\Back{B}[1]$, and $\Back{C}[1]$ in $W$.
Then by Theorem~\ref{thm:linear-factorizations}, reporting factorizations also takes $O(|W|)$ time.
\end{proof}

Combining this algorithm with Lemmas~\ref{lem:tiling-translation} and~\ref{lem:translation-bijection} yields the desired algorithmic result:

\begin{theorem}
\label{thm:linear-time-iso-enum}
Let $P$ be a polyomino with $n$ edges.
In $O(n)$ time, it can be determined if $P$ admits a tiling and the regular tilings of $P$ can be enumerated.
\end{theorem}

\section*{Acknowledgments}

The author thanks Stefan Langerman for fruitful discussions and comments that greatly improved the paper, and anonymous reviewers for pointing out an error in an earlier version of the paper.

\bibliographystyle{abbrv}
\bibliography{translation}

\end{document}